\documentclass[12pt]{article}
\usepackage{amsfonts}
\usepackage{amssymb,amsthm,latexsym}
\usepackage[numbers,compress]{natbib}
\usepackage{amsmath}
\usepackage{indentfirst}
\allowdisplaybreaks[4]
\newtheorem{theorem}{Theorem}[section]
\newtheorem{prop}[theorem]{Proposition}
\newtheorem{cor}[theorem]{Corollary}
\newtheorem{lemma}[theorem]{Lemma}

\newtheorem{remark}[theorem]{Remark}

\newtheorem{example}[theorem]{Example}

\numberwithin{equation}{section}

\begin{document}

\title{    Constant composition codes derived \\from  linear codes   }

\author{Long Yu{\thanks{Corresponding author.
\newline \indent ~~Email addresses:~longyu@mails.ccnu.edu.cn~(Long Yu), lxs6682@163.com~(Xiusheng Liu).},~~Xiusheng Liu }}
\date{  School of Mathematics and Physics, Hubei Polytechnic University, Huangshi,  435003, China}
\maketitle

\begin{abstract}
In this paper, we propose a class  of linear codes and obtain their weight distribution. Some of these codes are almost optimal. Moreover, several classes of constant composition codes(CCCs) are constructed as subcodes of  linear codes.
\end{abstract}



\section{Introduction}
Let $p$ be an odd prime and $q$ be a power of $p$. A linear $[n,k,d]$  code over the finite field $\mathbb{F}_q$ is a $k$-dimensional subspace of $\mathbb{F}_{q}^{n}$ with minimum Hamming distance $d$.  Let $D=\{d_1,d_2,\cdots,d_n\}\subseteq \mathbb{F}_q^*$, where $n$ is a positive integer. Let ${\rm Tr}$ denote the trace function from $\mathbb{F}_q$ to $\mathbb{F}_p$. We define a linear code of length $n$ over $\mathbb{F}_p$ by
 \begin{equation}\label{eq:linearcode}
 \mathcal{C}_D=\{c(a)=({\rm Tr}(ad_1),{\rm Tr}(ad_2),\cdots,{\rm Tr}(ad_n)|a\in \mathbb{F}_q \}.
 \end{equation}
Let $A_i$ denote the number of codewords with Hamming weight $i$ in a linear code $\mathcal{C}$ of length $n$. The weight enumerator of $\mathcal{C}$ is defined by
\[1+A_1X+A_2X^2+\cdots+A_{n}X^{n}.\]
The sequence $(1,A_1,\cdots,A_n)$ is called the weight distribution of the code $\mathcal{C}$.

The construction of linear code defined by (\ref{eq:linearcode}) is generic in the sense that many classes
of known codes could be produced by selecting the  suitable defining
set $D\subseteq \mathbb{F}_{q}$. So, the corrosponing exponential sums can be computed by some technologies of finite field. Therefore, the weight distributions of a large number of linear codes (cyclic codes) were obtained (see \cite{DingNiederreiter2007,DingDing2015,Liyueli2014,Liyue2014,LuoandFeng1,LuoandFeng2,yuliuamc,yu-liu2016,
Zeng2012,ZengHujiang2010,zhengwanghzeng2015,zhengwangyuliu2014,zhouding2014,ZhouLiFanH2016}, and references theirin).

Let $S=\{s_0,\cdots,s_{q-1}\}$ be an alphabet of size $q$.  An $[n,M,d,(\omega_0,\omega_1,\cdots,\omega_{q-1})_q]$ constant composition code(CCC) is a subset $C\subset S^n$ of size $M$, minimal distance $d$ and where the element $s_i$ occurs exactly $\omega_i$ times in each codeword in $C$.

Constant composition codes were studied
already in the 1960s. Both algebraic and combinatorial constructions of CCCs
have been proposed. For further information, the reader is referred to \cite{ChuColbournDukes2004,Ding2008,Dingyin2005,Dingyuan2005,LuoHellseth2011}.

The aim of this paper is to construct  CCCs from linear codes. Luo and Helleseth \cite{LuoHellseth2011} proposed a new way to obtain CCCs from some known cyclic codes. Recently, Yu and Liu \cite{yu-liuxiusheng2016} construct several classes of CCCs form linear codes.  Following this line, we   define a class of linear codes $\mathcal{C}_{D(\alpha)}$ by the set $D(\alpha)$. When $\alpha=0$, Ding and Ding \cite{DingDing2015} have already studied this kind of linear codes. So, for $\alpha\neq0$, we investigate the weight distribution of $\mathcal{C}_{D(\alpha)}$ (see Theorem~\ref{th:1}). Furthermore, we choose a kind of set $S_\gamma$ and obtain several classes of CCCs (see Theorem~\ref{th:2}).

\section{Preliminaries}

Throughout this paper, we let $q=p^m$, where $m$ is a positive integer. Let $\eta$ and $\overline{\eta}$ be the quadratic multiplicative character on $\mathbb{F}_q$ and $\mathbb{F}_p$, respectively. Let $\chi_1(\cdot)=\zeta_p^{{\rm {\rm Tr}}(\cdot)}$ and $\overline{\chi}_1=\zeta_p^{(\cdot)}$ be the canonical additive characters on $\mathbb{F}_q$ and $\mathbb{F}_p$, respectively. We define $\eta(0)=0=\overline{\eta}(0)$, then the quadratic Gaussian sum $G(\eta,\chi_1)$ on  $\mathbb{F}_q$ is defined by
\[G(\eta,\chi_1)=\sum_{x\in \mathbb{F}_q}\eta(x)\chi_1(x) ,\]
and  the quadratic Gaussian sum $G(\overline{\eta},\overline{\chi}_1)$ on  $\mathbb{F}_p$ is defined by
\[G(\overline{\eta},\overline{\chi}_1)=\sum_{x\in \mathbb{F}_p}\overline{\eta}(x)\overline{\chi}_1(x) .\]

The following  results are well known.
\begin{lemma}\cite{Lidl R}\label{lem:gauss}
Let the notations be given as above, we have
$$G(\eta,\chi_1)=(-1)^{m-1}\sqrt{-1}^{(\frac{p-1}{2})^2m} \sqrt{q}$$
and
$$G(\overline{\eta},\overline{\chi}_1)=\sqrt{-1}^{(\frac{p-1}{2})^2} \sqrt{p}.$$
\end{lemma}
 \begin{lemma}\cite{Lidl R}\label{lem:ercihanshuqiuhe}
Let $\chi$ be a nontrivial additive character of $\mathbb{F}_q$, and let $f(x)=a_2x^2+a_1x+a_0\in \mathbb{F}_q[x]$ with $a_2\neq0$. Then
\[\sum_{x\in \mathbb{F}_q}\chi(f(x))=\chi(a_0-a_1^2/(4a_2))\eta(a_2)G(\eta,\chi).\]
\end{lemma}

The conclusion of the following  lemma is easy to obtain.
\begin{lemma}\label{lem:ercitezheng}
If $m$ is odd, then $\eta(a)=\overline{\eta}(a)$ for any $a\in \mathbb{F}_p$. If $m$ is even,  then $\eta(a)=1$ for any $a\in \mathbb{F}_p^*$.
\end{lemma}


We will need the following lemma.
\begin{lemma}\cite{DingDing2015}\label{lem:changdu}
With the notations given as above. For each $\alpha\in \mathbb{F}_p$, let $$N_\alpha=\#\{x\in \mathbb{F}_{p^m}| {\rm Tr}(x^2)=\alpha\}.$$ Then
$$
N_\alpha=\left\{
  \begin{array}{ll}
    p^{m-1}, & \hbox{if $m$ is odd and $\alpha=0$;} \\
    p^{m-1}-(-1)^{(\frac{p-1}{2})^2\frac{m}{2}}(p-1)p^{\frac{m-2}{2}}, & \hbox{if $m$ is even and $\alpha=0$;} \\
     p^{m-1}+\overline{\eta}(-\alpha)(-1)^{(\frac{p-1}{2})^2(\frac{m+1}{2})}p^{\frac{m-1}{2}}, & \hbox{if $m$ is odd and $\alpha\neq0$;} \\
    p^{m-1}+(-1)^{(\frac{p-1}{2})^2\frac{m}{2}}p^{\frac{m-2}{2}}, & \hbox{if $m$ is even and $\alpha\neq0$.}
  \end{array}
\right.
$$
\end{lemma}

%

At the end of this section, we give the LFVC bound of constant composition code.
\begin{prop}\cite{LFVCbound2003}
Assume $nd-n^2+\omega_0^2+\omega_1^2+\cdots+\omega_{p-1}^2>0$. Then, an $[n,M,d,(\omega_\beta)_{\beta\in\mathbb{F}_p}]$ CCC satisfies the following inequality
\[M\leq nd/\left(nd-n^2+\omega_0^2+\omega_1^2+\cdots+\omega_{p-1}^2\right).\]
If $$M=nd/\left(nd-n^2+\omega_0^2+\omega_1^2+\cdots+\omega_{p-1}^2\right),$$ then we call CCC is optimal.
\end{prop}


%

\section{A class of linear code}
In this section, for a fixed $\alpha\in \mathbb{F}_p $, we define set
\[D(\alpha)=\{d\in\mathbb{F}_{p^m}^*|{\rm Tr}(d^2)=\alpha\}.\] The corresponding linear code is given as
\begin{equation}\label{eq:1}
   \mathcal{C}_{D(\alpha)}=\{c(a)=({\rm Tr}(ad_1),{\rm Tr}(ad_2),\cdots,{\rm Tr}(ad_{n_\alpha})|a\in \mathbb{F}_{p^m} \},
\end{equation}
where $n_\alpha$ is the length of $\mathcal{C}_{D(\alpha)}$. In particular, when $\alpha=0$, Ding and Ding \cite{DingDing2015} investigated the weight distribution of linear code $\mathcal{C}_{D(0)}$. Furthermore, Yu and Liu constructed a class of CCCs  from  code $\mathcal{C}_{D(0)}$. Here, we  calculate the weight distribution of linear code $\mathcal{C}_{D(\alpha)}$ for $\alpha\neq0$ and construct  CCCs form $\mathcal{C}_{D(\alpha)}$.

Let $\epsilon=(-1)^{(\frac{p-1}{2})^2\frac{m+1}{2}}$ for odd $m$ and $\tau=(-1)^{(\frac{p-1}{2})^2\frac{m}{2}}$ for even $m$, we have the following result.
\begin{lemma}\label{lem:zhongyaozhishuhe}
With the notations given as above. Then
\begin{eqnarray*}
   && \sum_{u\in \mathbb{F}_{p}^*}\zeta_p^{-u\alpha}\sum_{v\in \mathbb{F}_{p}^*}\sum_{x\in \mathbb{F}_{p^m}}\zeta_p^{{\rm Tr}(avx+ux^2)} \\
   &=& \left\{
     \begin{array}{ll}
       \varepsilon\overline{\eta}(-\alpha)(p-1)p^{\frac{m+1}{2}}  , & \hbox{if $m$ is odd and ${\rm Tr}(a^2)=0$;} \\
        \tau(p-1)p^{\frac{m}{2}} , & \hbox{if $m$ is even and ${\rm Tr}(a^2)=0$;}\\
-\epsilon p^{\frac{m+1}{2}}\left(\overline{\eta}({\rm Tr}(-a^2))+\overline{\eta}(-\alpha)\right), & \hbox{if $m$ is odd and ${\rm Tr}(a^2)\neq0$;} \\
       -\tau p^{\frac{m}{2}}\left((-1)^{(\frac{p-1}{2})^2}\overline{\eta}(\alpha {\rm Tr}({a^2}))p+1\right), & \hbox{if $m$ is even and ${\rm Tr}(a^2)\neq0$.}
     \end{array}
   \right.
\end{eqnarray*}

\end{lemma}

\begin{proof} By Lemma~\ref{lem:ercihanshuqiuhe}, we have
\begin{eqnarray}\label{eq:3.1}
\nonumber   & & \sum_{u\in \mathbb{F}_{p}^*}\zeta_p^{-u\alpha}\sum_{v\in \mathbb{F}_{p}^*}\sum_{x\in \mathbb{F}_{p^m}}\zeta_p^{{\rm Tr}(avx+ux^2)} \\
\nonumber   &=& G(\eta,\chi_1)\sum_{u\in \mathbb{F}_{p}^*}\zeta_p^{-u\alpha}\sum_{v\in \mathbb{F}_{p}^*}\eta(u)\zeta_p^{-{\rm Tr}(\frac{a^2v^2}{4u})}\\
   &=& G(\eta,\chi_1)\sum_{u\in \mathbb{F}_{p}^*}\eta(u)\zeta_p^{-u\alpha}\sum_{v\in \mathbb{F}_{p}^*}\zeta_p^{v^2{\rm Tr}(-\frac{a^2}{4u})}
\end{eqnarray}

If ${\rm Tr}(a^2)=0$, by Lemma~\ref{lem:ercitezheng}, one has
\begin{eqnarray*}
   & & G(\eta,\chi_1)\sum_{u\in \mathbb{F}_{p}^*}\eta(u)\zeta_p^{-u\alpha}\sum_{v\in \mathbb{F}_{p}^*}\zeta_p^{v^2{\rm Tr}(-\frac{a^2}{4u})}  \\
   &=& G(\eta,\chi_1)(p-1)\sum_{u\in \mathbb{F}_{p}^*}\eta(u)\zeta_p^{-u\alpha}\\
&=& G(\eta,\chi_1)(p-1)\eta(-\alpha)\sum_{u\in \mathbb{F}_{p}^*}\eta(u)\zeta_p^{u}\\
&=&\left\{
     \begin{array}{ll}
       G(\eta,\chi_1)(p-1)\eta(-\alpha)\sum\limits_{u\in \mathbb{F}_{p}^*}\overline{\eta}(u)\zeta_p^{u}, & \hbox{if $m$ is odd;} \\
       G(\eta,\chi_1)(p-1)\eta(-\alpha)\sum\limits_{u\in \mathbb{F}_{p}^*}\zeta_p^{u}, & \hbox{if $m$ is even.}
     \end{array}
   \right.\\
&=&\left\{
     \begin{array}{ll}
       G(\eta,\chi_1)(p-1)\overline{\eta}(-\alpha)G(\overline{\eta},\overline{\chi}_1)  , & \hbox{if $m$ is odd;} \\
       -G(\eta,\chi_1)(p-1) , & \hbox{if $m$ is even.}
     \end{array}
   \right.
\end{eqnarray*}

If ${\rm Tr}(a^2)\neq0$, by Lemmas~\ref{lem:ercihanshuqiuhe} and \ref{lem:ercitezheng}, then (\ref{eq:3.1}) is equal to
\begin{eqnarray*}
&& G(\eta,\chi_1)\sum_{u\in \mathbb{F}_{p}^*}\eta(u)\zeta_p^{-u\alpha}\sum_{v\in \mathbb{F}_{p}^*}\zeta_p^{v^2{\rm Tr}(-\frac{a^2}{4u})}\\
 &=& G(\eta,\chi_1)\sum_{u\in \mathbb{F}_{p}^*}\eta(u)\zeta_p^{-u\alpha}(\sum_{v\in \mathbb{F}_{p}}\zeta_p^{v^2{\rm Tr}(-\frac{a^2}{4u})}-1)\\
&=& G(\eta,\chi_1)\left(\sum_{u\in \mathbb{F}_{p}^*}\eta(u)\zeta_p^{-u\alpha}\overline{\eta}({\rm Tr}(-\frac{a^2}{4u}))G(\overline{\eta},\overline{\chi}_1)-\sum_{u\in \mathbb{F}_{p}^*}\eta(u)\zeta_p^{-u\alpha}\right)\\
&=&\left\{
     \begin{array}{ll}
       G(\eta,\chi_1)\left(G(\overline{\eta},\overline{\chi}_1)\overline{\eta}({\rm Tr}(-a^2))\sum\limits_{u\in \mathbb{F}_{p}^*}\zeta_p^{-u\alpha}-\overline{\eta}(-\alpha)\sum\limits_{u\in \mathbb{F}_{p}^*}\overline{\eta}(u)\zeta_p^{u}\right), & \hbox{if $m$ is odd;} \\
       G(\eta,\chi_1)\left(G(\overline{\eta},\overline{\chi}_1)\overline{\eta}(\alpha {\rm Tr}({a^2}))\sum\limits_{u\in \mathbb{F}_{p}^*}\overline{\eta}(-u\alpha)\zeta_p^{-u\alpha}-\sum\limits_{u\in \mathbb{F}_{p}^*}\zeta_p^{u}\right), & \hbox{if $m$ is even.}
     \end{array}
   \right.\\
&=&\left\{
     \begin{array}{ll}
       -G(\eta,\chi_1){\big(}G(\overline{\eta},\overline{\chi}_1)\overline{\eta}({\rm Tr}(-a^2))+
\overline{\eta}(-\alpha)G(\overline{\eta},\overline{\chi}_1){\big)}, & \hbox{if $m$ is odd;} \\
       G(\eta,\chi_1){\big(}G(\overline{\eta},\overline{\chi}_1)\overline{\eta}(\alpha {\rm Tr}({a^2}))G(\overline{\eta},\overline{\chi}_1)+1{\big)}, & \hbox{if $m$ is even.}
     \end{array}
   \right.
\end{eqnarray*}
The desired conclusions then follow from Lemma~\ref{lem:gauss}.
\end{proof}

The following lemma will be employed later.
\begin{lemma}
For $a\in \mathbb{F}_{p^m}^*$ and $\alpha\in \mathbb{F}_p$, let
 $$N(a)=|x\in\mathbb{F}_{p^m}|{\rm Tr}(x^2)=\alpha~~{\rm and}~~{\rm Tr}(ax)=0 |.$$
Then
\begin{eqnarray*}
     N(a)&=& \left\{
     \begin{array}{ll}
       p^{m-2}+\epsilon\overline{\eta}(-\alpha)p^{\frac{m-1}{2}}   , & \hbox{if $m$ is odd and ${\rm Tr}(a^2)=0$;} \\
        p^{m-2}+\tau p^{\frac{m}{2}-1} , & \hbox{if $m$ is even and ${\rm Tr}(a^2)=0$;}\\
p^{m-2}-\epsilon\overline{\eta}({\rm Tr}(-a^2))p^{\frac{m-3}{2}}, & \hbox{if $m$ is odd and ${\rm Tr}(a^2)\neq0$;} \\
       p^{m-2}-(-1)^{(\frac{p-1}{2})^2}\tau\overline{\eta}(\alpha {\rm Tr}({a^2}))p^{\frac{m}{2}-1}, & \hbox{if $m$ is even and ${\rm Tr}(a^2)\neq0$.}
     \end{array}
   \right.
\end{eqnarray*}
\end{lemma}
\begin{proof}

By definition, we have
\begin{eqnarray*}
  p^2N(a) &=&\sum_{x\in \mathbb{F}_{p^m}}(\sum_{u\in \mathbb{F}_{p}}\zeta_p^{u{\rm Tr}((x^2)-\alpha)})(\sum_{v\in \mathbb{F}_{p}}\zeta_p^{v{\rm Tr}(ax)})  \\
&=& \sum_{v\in \mathbb{F}_{p}^*}\sum_{x\in \mathbb{F}_{p^m}}\zeta_p^{v{\rm Tr}(ax)}+ \sum_{u\in \mathbb{F}_{p}}\sum_{x\in \mathbb{F}_{p^m}}\zeta_p^{u({\rm Tr}(x^2)-\alpha)}\\
   && +\sum_{u\in \mathbb{F}_{p}^*}\zeta_p^{-u\alpha}\sum_{v\in \mathbb{F}_{p}^*}\sum_{x\in \mathbb{F}_{p^m}}\zeta_p^{{\rm Tr}(avx+ux^2)} \\
&=&  pn_\alpha+\sum_{u\in \mathbb{F}_{p}^*}\zeta_p^{-u\alpha}\sum_{v\in \mathbb{F}_{p}^*}\sum_{x\in \mathbb{F}_{p^m}}\zeta_p^{{\rm Tr}(avx+ux^2)}.
\end{eqnarray*}
The desired conclusions then follow from Lemmas~\ref{lem:changdu} and \ref{lem:zhongyaozhishuhe}.

\end{proof}

Now, we give the main result in this section.
\begin{theorem}\label{th:1}
Let the notations be given as above.
\begin{itemize}
  \item For odd $m$,   $\mathcal{C}_{D(\alpha)}$ defined by (\ref{eq:1}) is an $[p^{m-1}+\overline{\eta}(-\alpha)\epsilon p^{\frac{m-1}{2}},m]$ code with weight distribution in Table~$1$.
\begin{table}[!h]
\tabcolsep 2pt
\caption{For odd $m$, the weight distribution of $\mathcal{C}_{D(\alpha)}$}
\vspace*{0pt}
\begin{center}
\begin{tabular}{|c|c|}
  \hline
  weight & frequency  \\ \hline
$0$ & $1$  \\ \hline
   $(p-1)p^{m-2}$ & $p^{m-1}-1$  \\ \hline
  $(p-1)p^{m-2}+\epsilon\left(\overline{\eta}(-1)+p
\overline{\eta}(-\alpha)\right)p^{\frac{m-3}{2}}$& $\frac{p-1}{2}\left( p^{m-1}+\overline{\eta}(-1)\epsilon p^{\frac{m-1}{2}}\right)$  \\ \hline
  $(p-1)p^{m-2}+\epsilon\left(-\overline{\eta}(-1)+p
\overline{\eta}(-\alpha)\right)p^{\frac{m-3}{2}}$ & $\frac{p-1}{2}\left(p^{m-1}-\overline{\eta}(-1)\epsilon p^{\frac{m-1}{2}}\right)$  \\
  \hline
\end{tabular}
\end{center}
\end{table}
where  $\epsilon=(-1)^{(\frac{p-1}{2})^2\frac{m+1}{2}}$.
  \item For even $m$, $\mathcal{C}_{D(\alpha)}$ defined by (\ref{eq:1}) is an $[ p^{m-1}+\tau p^{\frac{m-2}{2}},m]$ code with weight distribution in Table~$2$.
\begin{table}[!h]
\tabcolsep 2pt
\caption{For even $m$, the weight distribution of $\mathcal{C}_{D(\alpha)}$}
\vspace*{0pt}
\begin{center}
\begin{tabular}{|c|c|}
  \hline
 weight & frequency  \\ \hline
 $0$ & $1$  \\ \hline
  $(p-1)p^{m-2}$ & $\frac{p+1}{2}p^{m-1}-\tau\frac{p-1}{2}p^{\frac{m}{2}-1}-1 $ \\ \hline
  $(p-1)p^{m-2}+2\tau p^{\frac{m}{2}-1} $& $\frac{p-1}{2}\left( p^{m-1}+\tau p^{\frac{m}{2}-1}\right)$  \\  \hline
\end{tabular}
\end{center}
\end{table}
where $\tau=(-1)^{(\frac{p-1}{2})^2\frac{m}{2}}$.
\end{itemize}
\end{theorem}
\begin{proof}
It is easy to obtain $n_\alpha $  from Lemma~\ref{lem:changdu}.
If $m$ is odd, for $a\neq0$, then
\begin{eqnarray*}
   && wt(c(a))=n_\alpha-N(a) \\
   &=& \left\{
                          \begin{array}{ll}
                            (p-1)p^{m-2}, & \hbox{if ${\rm Tr}(a^2)=0$;} \\
                           (p-1)p^{m-2}+\epsilon\left(\overline{\eta}({\rm Tr}(-a^2))+p
\overline{\eta}(-\alpha)\right)p^{\frac{m-3}{2}}, & \hbox{if ${\rm Tr}(a^2)\neq0$.}
                          \end{array}
                        \right.
\end{eqnarray*}

Note that ${\rm Tr}(a^2)\in \mathbb{F}_p^*$ is a square element, then, by Lemma~\ref{lem:changdu}, we have the number of $a$ is
$$(p-1)\left(p^{m-1}+ \overline{\eta}(-1)\epsilon p^{\frac{m-1}{2}}\right).$$

Note that ${\rm Tr}(a^2)\in \mathbb{F}_p^*$ is a non-square element, then, by Lemma~\ref{lem:changdu}, we get the number of $a$ is
$$(p-1)\left(p^{m-1}- \overline{\eta}(-1)\epsilon p^{\frac{m-1}{2}}\right).$$

If $m$ is even, then
\begin{eqnarray*}
   && wt(c(a))=n_\alpha-N(a) \\
   &=& \left\{
                          \begin{array}{ll}
                            (p-1)p^{m-2}, & \hbox{if ${\rm Tr}(a^2)=0$;} \\
                           (p-1)p^{m-2}+\tau\left(1+(-1)^{(\frac{p-1}{2})^2}\overline{\eta}(\alpha)\right)p^{\frac{m}{2}-1}, & \hbox{if ${\rm Tr}(a^2)$ is square;} \\
       (p-1)p^{m-2}+\tau\left(1-(-1)^{(\frac{p-1}{2})^2}\overline{\eta}(\alpha)\right)p^{\frac{m}{2}-1}, & \hbox{if ${\rm Tr}(a^2)$ is non-square.}
                          \end{array}
                        \right.
\end{eqnarray*}
Note that when $a\in \mathbb{F}_{p^m}^*$, we have $wt(c(a))>0$ for any $m$. This implies that the dimension of linear code $\mathcal{C}_{D(\alpha)}$ is $m$.

\begin{example}
Let $p=5$, $m=3$ and $\alpha$ is  quare.  By Magma, we have $\mathcal{C}_{D(\alpha)}$ is a $[30,3,20]$ code, whcih confirms the results in Table~$1$.

Let $p=7$, $m=3$ and $\alpha$ is  non-square. By Magma, we have $\mathcal{C}_{D(\alpha)}$ is a $[56,3,42]$ code, whcih agrees with the results in  Table~$1$.

Let $p=3$, $m=6$. By Magma, we have $\mathcal{C}_{D(\alpha)}$ is a $[234,6,144]$ code, whcih agrees with the results in Table~$2$.

Let $p=3$, $m=4$. By Theorem~\ref{th:1}, we have $\mathcal{C}_{D(\alpha)}$ is a $[30,4,18]$ code,  whcih confirms the results in Table~$2$. As we known, for $n=30$, $k=4$, the best code has parameters $[30,4,19]$. This implies that $\mathcal{C}_{D(\alpha)}$ is almost optimal.

\end{example}


\end{proof}
%
%
%

\section{Constant composition codes}
In this section, we will construct a class of CCCs as subcodes of linear code $\mathcal{C}_{D(\alpha)}$ defined by (\ref{eq:1}).
For each $\gamma\in\mathbb{F}_{p} $, we let
$$S_\gamma=\{a\in\mathbb{F}_{p^m}^*|{\rm Tr}(a^2)=\gamma\}.$$
Let $\alpha\in\mathbb{F}_{p}^* $. Define
\[\mathcal{C}_{D(\alpha)}^\gamma=\{c(a)|a\in S_\gamma\}.\]
\begin{theorem}\label{th:2}
Let $m$ be even and $\tau=(-1)^{(\frac{p-1}{2})^2\frac{m}{2}}$. The code $\mathcal{C}_{D(\alpha)}^\gamma$ is a CCC with parameters $(n_\alpha,M,d,[\omega_\beta]_{\beta\in\mathbb{F}_{p}} )$, where

1) in the case of $\gamma=0$:
\begin{eqnarray*}
  n_\alpha &=& p^{m-1}+\tau p^{\frac{m-2}{2}} \\
  M &=& p^{m-1}-\tau (p-1)p^{\frac{m-2}{2}}-1 \\
  \omega_\beta &=&\left\{
     \begin{array}{ll}
       p^{m-2}+\tau p^{\frac{m}{2}-1}, & \hbox{$\beta=0$;} \\
       p^{m-2}, & \hbox{$\beta\neq0$.}
     \end{array}
   \right. \\
  d &=& \left\{
          \begin{array}{ll}
              (p-1)p^{m-2}-2p^{\frac{m-2}{2}}, & \hbox{ $\tau=-1$;} \\
            (p-1)p^{m-2}, & \hbox{ $\tau=1$.}
          \end{array}
        \right.
\end{eqnarray*}

2) in the case of square $\alpha\gamma$:
\begin{eqnarray*}
  n_\alpha &=& p^{m-1}+\tau p^{\frac{m-2}{2}}; \\
  M &=& p^{m-1}+\tau p^{\frac{m-2}{2}}; \\
  \omega_\beta &=&\left\{
      \begin{array}{ll}
 p^{m-2}+(-1)^{(\frac{p-1}{2})^2}\tau p^{\frac{m}{2}-1}, & \hbox{$\beta=0$;} \\
        p^{m-2}, & \hbox{$\beta=\pm \sqrt{\alpha \gamma}$;} \\
        p^{m-2}+(-1)^{(\frac{p-1}{2})^2}\tau\overline{\eta}( \alpha \gamma-\beta^2)p^{\frac{m}{2}-1}, & \hbox{otherwise.}
      \end{array}
    \right. \\
 d &=& \left\{
          \begin{array}{ll}
              (p-1)p^{m-2}-2p^{\frac{m-2}{2}}, & \hbox{ $\tau=-1$;} \\
            (p-1)p^{m-2}, & \hbox{ $\tau=1$.}
          \end{array}
        \right.
\end{eqnarray*}

3) in the case of non-square $\alpha\gamma$:
\begin{eqnarray*}
  n_\alpha &=& p^{m-1}+\tau p^{\frac{m-2}{2}}; \\
  M &=& p^{m-1}+\tau p^{\frac{m-2}{2}}; \\
  \omega_\beta &=&\left\{
                    \begin{array}{ll}
                       p^{m-2}-(-1)^{(\frac{p-1}{2})^2}\tau p^{\frac{m}{2}-1}, & \hbox{$\beta=0$;} \\
                      p^{m-2}+(-1)^{(\frac{p-1}{2})^2}\tau\overline{\eta}( \alpha \gamma-\beta^2)p^{\frac{m}{2}-1} , & \hbox{otherwise.}
                    \end{array}
                  \right.
\\
  d &=& \left\{
          \begin{array}{ll}
              (p-1)p^{m-1}-2p^{\frac{m-2}{2}}, & \hbox{ $\tau=-1$;} \\
            (p-1)p^{m-1}, & \hbox{ $\tau=1$.}
          \end{array}
        \right.
\end{eqnarray*}

%
%
\end{theorem}
\begin{proof}
Note that $\alpha\neq0$, for any $\beta\in \mathbb{F}_{p}$, by Lemmas  \ref{lem:ercihanshuqiuhe} and \ref{lem:ercitezheng}, we have
\begin{eqnarray}\label{eq:4.1}
\nonumber  \omega_\beta &=& \frac{1}{p}\sum_{x\in {D(\alpha)}}\sum_{u\in \mathbb{F}_{p}}\zeta_p^{u({\rm Tr}(ax)-\beta)} \\
\nonumber    &=& \frac{1}{p^2}\sum_{x\in \mathbb{F}_{p^m}}\sum_{u\in \mathbb{F}_{p}}\zeta_p^{u({\rm Tr}(ax)-\beta)}\sum_{v\in \mathbb{F}_{p}}\zeta_p^{v({\rm Tr}(x^2)-\alpha)} \\
 \nonumber   &=& \frac{1}{p^2}\sum_{u\in \mathbb{F}_{p}}\sum_{v\in \mathbb{F}_{p}}\zeta_p^{-u\beta}\zeta_p^{-v\alpha}\sum_{x\in \mathbb{F}_{p^m}}\zeta_p^{{\rm Tr}(aux+vx^2)}\\
\nonumber&=& p^{m-2}+\frac{1}{p^2}\sum_{u\in \mathbb{F}_{p}^*}\zeta_p^{-u\beta}\sum_{x\in \mathbb{F}_{p^m}}\zeta_p^{{\rm Tr}(aux)}+\frac{1}{p^2}\sum_{v\in \mathbb{F}_{p}^*}\zeta_p^{-v\alpha}\sum_{x\in \mathbb{F}_{p^m}}\zeta_p^{{\rm Tr}(vx^2)}\\
\nonumber& &+\frac{1}{p^2}\sum_{u\in \mathbb{F}_{p}^*}\sum_{v\in \mathbb{F}_{p}^*}\zeta_p^{-u\beta}\zeta_p^{-v\alpha}\sum_{x\in \mathbb{F}_{p^m}}\zeta_p^{{\rm Tr}(aux+vx^2)}\\
\nonumber&=& p^{m-2}-\frac{1}{p^2}G(\eta,\chi_1)+\frac{1}{p^2}\sum_{u\in \mathbb{F}_{p}^*}\sum_{v\in \mathbb{F}_{p}^*}\zeta_p^{-u\beta}\zeta_p^{-v\alpha} \zeta_p^{{\rm Tr}( -\frac{a^2u^2}{4v})}G(\eta,\chi_1)\\
\nonumber&=& p^{m-2}-\frac{1}{p^2}G(\eta,\chi_1)+\frac{1}{p^2}G(\eta,\chi_1)\sum_{v\in \mathbb{F}_{p}^*}\zeta_p^{-v\alpha}(\sum_{u\in \mathbb{F}_{p}}\zeta_p^{-\beta u- \frac{{\rm Tr}(a^2)}{4v}u^2}-1)\\
&=& p^{m-2}+\frac{1}{p^2}G(\eta,\chi_1)\sum_{v\in \mathbb{F}_{p}^*}\zeta_p^{-v\alpha}\sum_{u\in \mathbb{F}_{p}}\zeta_p^{-\beta u- \frac{{\rm Tr}(a^2)}{4v}u^2}.
\end{eqnarray}
If ${\rm Tr}(a^2)=0$, i.e. $\gamma=0$, then (\ref{eq:4.1}) is equal to
\begin{eqnarray*}
  \omega_\beta
&=& p^{m-2}+\frac{1}{p^2}G(\eta,\chi_1)\sum_{v\in \mathbb{F}_{p}^*}\zeta_p^{-v\alpha}\sum_{u\in \mathbb{F}_{p}}\zeta_p^{-\beta u }\\
&=&\left\{
     \begin{array}{ll}
       p^{m-2}-\frac{1}{p}G(\eta,\chi_1), & \hbox{$\beta=0$;} \\
       p^{m-2}, & \hbox{$\beta\neq0$.}
     \end{array}
   \right.
\end{eqnarray*}
If ${\rm Tr}(a^2)\neq0$, i.e. $\gamma\neq0$, then (\ref{eq:4.1}) is equal to
\begin{eqnarray*}
 \omega_\beta &=& p^{m-2}+\frac{1}{p^2}G(\eta,\chi_1)G(\overline{\eta},\overline{\chi}_1)\sum_{v\in \mathbb{F}_{p}^*}\zeta_p^{-v\alpha}\overline{\eta}(-\frac{{\rm Tr}(a^2)}{4v})\zeta_p^{\frac{\beta^2v}{{\rm Tr}(a^2)}}
\\
&=& p^{m-2}+\frac{1}{p^2}G(\eta,\chi_1)G(\overline{\eta},\overline{\chi}_1)\overline{\eta}(-{\rm Tr}(a^2))\sum_{v\in \mathbb{F}_{p}} \overline{\eta}(v)\zeta_p^{(-\alpha+\frac{\beta^2}{{\rm Tr}(a^2)})v}.
\end{eqnarray*}
Case I: if $\alpha \gamma$ is a square element, then we have
\begin{eqnarray*}
 \omega_\beta
&=& \left\{
      \begin{array}{ll}
        p^{m-2}, & \hbox{$\beta=\pm \sqrt{\alpha \gamma}$;} \\
        p^{m-2}+\frac{1}{p^2}G(\eta,\chi_1)G(\overline{\eta},\overline{\chi}_1)\overline{\eta}(-{\rm Tr}(a^2))\overline{\eta}(-\alpha+\frac{\beta^2}{{\rm Tr}(a^2)})\sum\limits_{v\in \mathbb{F}_{p}} \overline{\eta}(v)\zeta_p^{v}, & \hbox{otherwise.}
      \end{array}
    \right.
\\
&=&\left\{
      \begin{array}{ll}
        p^{m-2}, & \hbox{$\beta=\pm \sqrt{\alpha \gamma}$;} \\
        p^{m-2}+\frac{1}{p^2}G(\eta,\chi_1)G^2(\overline{\eta},\overline{\chi}_1)\overline{\eta}( \alpha {\rm Tr}(a^2)-\beta^2), & \hbox{otherwise.}
      \end{array}
    \right.
\end{eqnarray*}
Case II: if $\alpha \gamma$ is a non-square element, then we have
\begin{eqnarray*}
 \omega_\beta
&=&   p^{m-2}+\frac{1}{p^2}G(\eta,\chi_1)G(\overline{\eta},\overline{\chi}_1)\overline{\eta}(-{\rm Tr}(a^2))\overline{\eta}(-\alpha+\frac{\beta^2}{{\rm Tr}(a^2)})\sum\limits_{v\in \mathbb{F}_{p}} \overline{\eta}(v)\zeta_p^{v}\\
&=& p^{m-2}+\frac{1}{p^2}G(\eta,\chi_1)G^2(\overline{\eta},\overline{\chi}_1)\overline{\eta}( \alpha {\rm Tr}(a^2)-\beta^2).
\end{eqnarray*}

$n_\alpha$ is the length of  linear code $\mathcal{C}_{D(\alpha)}$. Note that $M$ is the size of $S_\gamma$, which can be obtained from Lemma~\ref{lem:changdu}.
Denote by $d_H(c(a_1),c(a_2))$ the Hamming distance of $c(a_1)$ and $(a_2)$. When $a_1$ and $a_2$ run through $S_\gamma$ with $a_1\neq a_2$, then $a_1- a_2$ runs through $\mathbb{F}_{p^m}^*$. Therefore, the minimal distance of $\mathcal{C}_{D(\alpha)}^\gamma$ is the same as that of $\mathcal{C}_{D(\alpha)}$.

The desired conclusions then follow from Lemmas~\ref{lem:gauss} and \ref{th:1}.
\end{proof}
\begin{cor}\label{cor1}
Let $t\in \mathbb{F}_p^*$, then
\[\sum_{x \in \mathbb{F}_p} \overline{\eta}(t-x^2)=(-1)^{(\frac{p-1}{2})^2}.\]
\end{cor}
\begin{proof}
Note that $$\sum_{\beta\in \mathbb{F}_p}\omega_\beta=n_\alpha.$$ By   Theorem~\ref{th:2}, we finish the proof.
\end{proof}
Therefore, from Theorem~\ref{th:2} and Corollary~\ref{cor1}, we have the following result.
\begin{prop}\label{pro:1}
For $\gamma\neq0$, then
\[\sum_{\beta\in \mathbb{F}_p}\omega_\beta^2=p^{2m-3}+p^{m-1}+2\tau p^{\frac{3m}{2}-3}.\]
\end{prop}
\begin{remark}
By Proposition~\ref{pro:1} and Theorem~\ref{th:2}, We can check that $$n_\alpha d-n_\alpha^2+\omega_0^2+\omega_1^2+\cdots+\omega_{p-1}^2\leq0, {\rm ~~~for~~} \tau=\pm 1.$$   Therefore, the LFVC bound cannot be applied to measure the optimality of these CCCs.
\end{remark}

\textbf{Acknowledgment}

The authors would like to thank the referees for their comments that improved the readability of the paper.
The work of L. Yu was support by  research
funds of HBPU(Grant No. 17xjz04R).



\end{document}